\documentclass[11pt,letterpaper]{article}

\usepackage{amsmath,amssymb,amsfonts,amsthm}
\usepackage{fancybox}
\usepackage[noend,noline,ruled,scleft,nofillcomment]{algorithm2e}
\usepackage{caption} 
\usepackage{hyperref,fullpage}
\usepackage{tikz-network}
\usepackage{tikz}
\usetikzlibrary{calc, hobby}
\usetikzlibrary{decorations.markings}
\usetikzlibrary{arrows,petri,topaths}
\usepackage{color}

\newtheorem{theorem}[equation]{Theorem}
\newtheorem{lemma}{Lemma}
\newtheorem{corollary}{Corollary}

\newtheorem{fact}{Fact}
\newtheorem{observation}{Observation}
\theoremstyle{definition}
\newtheorem{definition}{Definition}

\include{amssymb}

\newcommand{\etal}{{\em et al. }}

\newcommand{\eps}{\varepsilon}

\newcommand{\COMMENTED}[1]{{}}

\newcommand{\bind}{{\mathrm{bind}}}

\begin{document}

\title{A Hall-type theorem with algorithmic consequences in planar graphs }

\author{Ebrahim Ghorbani, \;\; Hossein Jowhari\footnote{The author is also affiliated with the School of Mathematics at the Institute for Research in Fundamental Sciences (IPM) in Tehran, Iran. This work is supported by IPM under
Project Number 98050014. } }
\date{
\footnotesize
Faculty of Mathematics, K. N. Toosi University of Technology,\\
P. O. Box 16765-3381, Tehran, Iran.\\
{\tt \{ghorbani,jowhari\}@kntu.ac.ir}}
\maketitle

\begin{abstract}
Given a graph $G=(V,E)$, for a vertex set $S\subseteq V$, let $N(S)$ denote
the set of vertices in $V$ that have a neighbor in $S$.
Extending the concept of binding number of graphs by Woodall~(1973), for a vertex set $X \subseteq V$, we define the binding number of $X$, denoted by $\bind(X)$, as the maximum number $b$ such that for every $S \subseteq X$  where $N(S)\neq V(G)$  it holds that 
 $|N(S)|\ge b {|S|}$. Given this definition, we prove that if a graph $V(G)$ contains a subset $X$  with $\bind(X)= 1/k$ where $k$ is an integer, then $G$ possesses a matching of size at least $|X|/(k+1)$. 
Using this statement, we derive tight bounds for the estimators of
the matching size in planar graphs. 
These estimators are previously used in designing sublinear space algorithms for approximating the maching size in the data stream model of computation.
In particular, we show that
the number of locally superior vertices is a $3$
factor approximation of the matching size in planar graphs. The previous analysis by Jowhari (2023) proved a $3.5$ approximation factor. As another application, we show
a simple variant of an estimator by Esfandiari \etal (2015) achieves $3$ factor approximation 
of the matching size in planar graphs. Namely, let $s$ be the number of edges with
both endpoints having degree at most $2$ and let $h$ be the number of 
vertices with degree at least $3$. We prove that when the graph is planar, the size of 
matching is at least $(s+h)/3$. This result generalizes a known fact that every planar graph on $n$ vertices with minimum degree $3$ has a matching of size at least $n/3$. 
\end{abstract}

\section{Introduction} 
In the case of bipartite graphs $G=(A\cup B,E)$, the celebrated Hall's theorem \cite{Hall} formulates the necessary and sufficient 
conditions for the existence of a perfect matching on $A$ in terms of the size of the neighborhoods of the subsets of $A$. Namely, Hall's theorem states that if for every $S \subseteq A$ it holds that $|N(S)| \ge |S|$, then $\nu(G) \ge |A|$ where $\nu(G)$ is size of a maximum matching in $G$.   For general graphs, perhaps the most relevant result is due to Woodall \cite{Woodall73}. In his seminal work, Woodall defines the binding number of the graph $G$, denoted by $\bind(G)$, as the largest number $b$ such that for every $S \subseteq V(G)$, where $N(S) \neq V(G)$, it holds that $|N(S)| \ge b|S|$. Based on this definition, Woodall showed the following result. 

\begin{theorem}	[Woodal \cite{Woodall73}] \label{thm:woodall}
Let $G$ be a graph with $\bind(G) = c \in [0,\frac12]$. Then  $\nu(G) \ge \frac{c}{c+1}|V(G)|$.
\end{theorem}  

\subsection{Main result}
 In this paper, we utilize a similar concept and prove a statement that can be thought of as another variant of Hall's theorem for general graphs. For this, we extend the defintion of Woodall's binding number to arbitrary subsets of vertices. 

\begin{definition} 
\label{def:binding}
Let $G=(V,E)$ be a simple graph. The binding number of the vertex set $X \subseteq V$, denoted
by $\bind(X)$, is the largest number $b$ such that for every subset $S \subseteq X$ where $N(S) \neq V(G)$, we have $|N(S)| \ge b |S|$. 
\end{definition}

The following is the main result of this paper.

\begin{theorem} 
\label{thm:main}
Let $k \ge 2$ be an integer. Let $X$ be a vertex set in graph $G$ with binding number $c = \frac1k$.
Then $\nu(G) \ge \frac{c}{c+1}|X|$. 
\end{theorem}

We prove Theorem~\ref{thm:main} using max-flow min-cut theorem as well as lower bounds for matching size in trees and unicyclic graphs with bounded degree. Unfortunately our proof strategy only works when $\bind(X)= \frac1k$ for an integer $k$. We conjecture that a similar statement should hold for all values $\bind(X) \in [0,\frac12]$, and  leave it as an open question. It is worth noting that Theorem~\ref{thm:main} does not immediately follow from Theorem~\ref{thm:woodall}. Taking the induced subgraph on $X \cup N(X)$ is not sufficient, as its binding number could be much smaller than $\bind(X)$. On the other hand, it is not clear how to apply Woodall's arguments, which are based on the Tutte-Berge formula \cite{TB58}, to derive our result from Theorem~\ref{thm:woodall}.

\subsection{Applications}

We  now explore some applications of Theorem \ref{thm:main}.

\paragraph{Locally superior vertices.} Given an undirected graph $G=(V,E)$, the vertex $u \in V$ is called locally superior if $u$ has an adjacent vertex with degree at most $\deg(u)$. The set of locally superior vertices of $G$ is denoted by $L(G)$. The aboricity of the undirected graph $G$ is the minimum number of forests that cover all edges of $G$.  In \cite{Jowhari23}, it was shown that the number of locally superior vertices approximates the size of the matching in graphs with bounded arboricity. In particular, it was shown that when $G$ is planar, we have $\nu(G) \le |L(G)| \le 3.5 \nu(G)$. 
Moreover, when the arboricity of $G$ is bounded by $\alpha$, it was shown that 
$\nu(G) \le |L(G)| \le (\alpha+2) \nu(G)$.  
By utilizing Theorem \ref{thm:main}, we tighten the previous analysis and establish the following theorem. The proof appears in Section \ref{sec:consequences}.

\begin{theorem}
\label{thm:ls}
Let $G$ be a graph and let $\alpha \ge 2$ be an integer. We have $\nu(G) \le |L(G)| \le 3 \nu(G)$ when $G$ is planar. 
When the arboricity of $G$ is bounded by $\alpha$, we have $\nu(G) \le |L(G)| \le (\alpha+1)\nu(G)$
\end{theorem}
As an immediate consequence, this shows that the approximation factor of the algorithm presented in \cite{Jowhari23} for planar graphs is in fact $3\pm\epsilon$. 
In a similar fashion, this leads to improved bounds for approximating $\nu(G)$ when $G$ is a bounded arboricity graph in the vertex-arrival streams. In the vertex-arrival model (also known as the adjacency list model), the input graph is given as a sequence of vertices in arbitrary order. Each vertex is also accompanied with a list of its neighbors. In the edge-arrival model (also known as the arbitrary order model), the input graph is given as a sequence of edges in arbitrary order. 

\paragraph{Another $3$ factor approximation.} The work by 
Esfandiari \etal \cite{EHLMO15} were the first to demonstrate that the size of the maximum matching in bounded arboricity graphs can be approximated within a constant factor by considering only the local neighborhood of the vertices and edges of the graph. More precisely, given a simple graph $G$, let $H_{t}(G)$ denote the set of vertices with degree greater than
$t$ and let $S_t(G)$ denote the set of edges in $G$ with both
endpoints having degree
at most $t$. By applying an extension of the Hall's theorem for bipartite graphs (see  \cite[Lemma~3.1]{EHLMO15}), Esfandiari \etal showed that 
$\nu(G) \le |S_{(2\alpha+3)}(G)|+|H_{(2\alpha+3)}(G)| \le (5\alpha+9)\nu(G)$ when the arboricity of $G$ is bounded by $\alpha$. In particular, 
since the arboricity of planar graphs is bounded by $3$, setting $\alpha=3$, this estimator gives a $24$ factor approximation of the matching size in planar graphs.
As another application
of Theorem \ref{thm:main}, we prove the following inequality which gives a much better estimation of the matching size in planar graphs. 
\begin{theorem}
\label{thm:shallowhigh}
For a planar graph $G$, we have 
$\nu(G) \le |S_2(G)| + |H_2(G)| \le 3\nu(G)$. 
\end{theorem}

\paragraph{Matchings in planar graphs with minimum degree $3$.} It is known that a planar graph
on $n$ vertices with minimum degree $3$ has a matching of size at least $n/3$. This fact was first discovered by Nishizeki and Baybars \cite{NishizekiB79}. The proof in \cite{NishizekiB79} is based on the famous Tutte--Berge formula \cite{TB58}. A different proof, again based on Tutte--Berge formula, for 3-connected planar graphs is given in \cite{BiedlDDFK04}. Biedl \cite{B19} presented 
a linear-time algorithm for finding a matching of size $n/3$ in planar graph using an alternative proof.
A straightforward consequence of Theorem~\ref{thm:shallowhigh} leads to another proof of this fact, which might be of independent interest. It is worth noting that in a graph $G$ with a minimum degree of $3$, the set $S_2(G)$ is empty.

\section{Proof of Theorem \ref{thm:main}}
We begin by stating two lemmas that are used in our subsequent arguments.

\begin{lemma} 
\label{lem:richness} Let  $G$ a simple graph and $X \subseteq V(G)$ be a set of vertices with $\bind(X)=1/k$ for some integer $k \ge 1$. 
Then there exists a function $f:X \rightarrow N(X)$ with the following properties:
\begin{itemize}
\item for all $x \in X$, $f(x)$ is an adjacent vertex of $x$,
\item for all $y \in N(X)$, $|\{x : \; f(x) = y\}| \le k$.
\end{itemize}
\end{lemma}

\begin{proof} The proof is a straightforward generalization of the proof of Hall's theorem using the celebrated max-flow min-cut theorem. We construct a single-source, single-destination flow network $F$ as follows. The vertex set of $F$ is $A \cup B \cup \{s,t\}$ where $s$ and $t$ are the source and the sink vertices, respectively. For each $x$ in $X$, we put a corresponding vertex in $A$. Likewise, for each $y \in N(X)$, we put a corresponding vertex in $B$. For each $x \in A$ and $y \in B$, we add the edge $(x,y)$ if and only if $y$ is an adjacent vertex of $x$ in $G$. The capacity of these edges is set to $\infty$  (to be precise, a value greater than $|X|$ suffices here). We add an edge from the source $s$ to all the vertices in $A$ with capacity $1$. Finally, we add an edge from every vertex in $B$ to the sink $t$ with capacity $k$. 

We claim the value of the maximum $s$-$t$ flow in $F$ is $|X|$. This claim implies the existence of an integral $s$-$t$ flow of value $|X|$. As a result, the flow on the edges between  $A$ and $B$ gives us the assignment function $f$ with the desired properties.

With regard to our maximum flow claim, suppose the $s$-$t$ cut $(S,T)$, where $S = \{s\}\cup A'\cup B'$ and $T= A''\cup B''\cup \{t\}$, is a minimum $s$-$t$ cut in $F$. Note that here  $A' \cup A'' = A$ and $B' \cup B'' = B$. By the max-flow min-cut theorem, the capacity of the $(S,T)$ cut equals the value of the maximum flow which is at most $|X|$. 
From this, we conclude that $(S,T)$ should not cut any of the infinity edges between $A$ and $B$. Hence the capacity of the cut is exactly $|A''| + k|B'|$. Also by the same observation, we have $N(A') \subseteq B'$. Hence
$$c(S,T) = |A''| + k|B'| \ge |A''| + k |N(A')| .$$
If $N(A') = V(G)$, then $|B'| = |V(G)|$. In this case, clearly $c(S,T) \ge k|B'| \ge |X|$.  Therefore, we may safely assume that $N(A') \neq V(G)$. By the fact that $X$ is a vertex set with binding number $1/k$, we obtain
$$c(S,T) \ge |A''| + k |N(A')| \ge |A''| + |A'| = |A| = |X|.  $$
Consequently, the value of the maximum $s$-$t$ flow is $|X|$, as claimed. This finishes the proof.
\end{proof}

\begin{lemma} 
\label{lem:max3}
Let $k\ge 3$. The following statements are true.
\begin{itemize}
\item Let $T$ be a tree of order $n$ with maximum degree $k$. Then $\nu(T) \ge \frac{n-1}k.$
\item Let $G$ be a connected graph of order $n$ that has exactly one cycle and its maximum degree is 
$k$. We have $\nu(G) \ge \frac{n}k$. Here, $G$ may have parallel edges, in which case, we count the only pair of parallel edges as a cycle.
\end{itemize}
\end{lemma}

\begin{proof}
	 The first part follows from a more general fact (see \cite[p.~121]{West01}) that in a  bipartite graph $G$ with $m$ edges and maximum degree $k$, we have $\nu(G) \ge m/k$ .  For the second part, we provide a proof by  induction on the number of vertices. Consider a connected graph $G$ of order $n$ that has only one cycle and its maximum degree is $k \ge 3$. The base case where $n=2$ is clearly true. Note that we may view $G$ as a collection of trees $T_1, \ldots, T_r$ that are attached to a simple cycle $C$. See Figure \ref{fig:cycle:trees}\,(a).

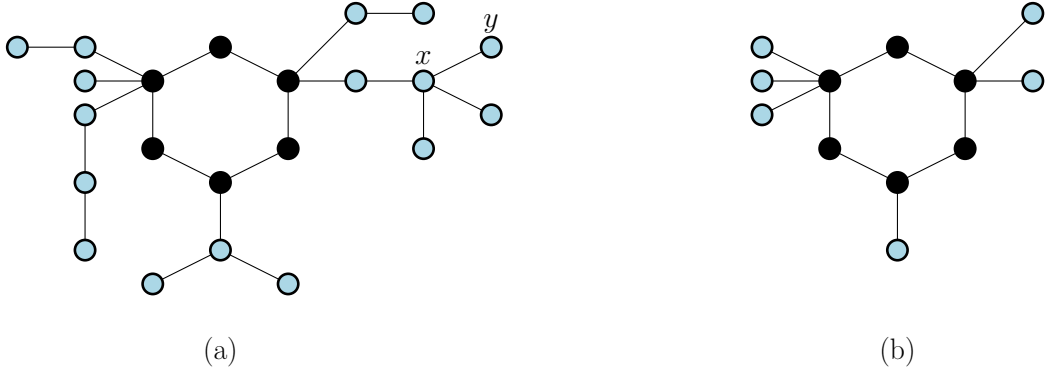
\begin{figure}
\centering
\begin{tikzpicture}[scale=0.45,transform shape]
\Text[x=8, y = 8.7]{\Huge $x$}
\Text[x=10, y = 9.7]{\Huge $y$}
\Vertex[x=2,y=5, color=black]{a1}
\Vertex[x=0,y=6, color=black]{a2}
\Vertex[x=0,y=8, color=black]{a3}
\Vertex[x=2,y=9, color=black]{a4}
\Vertex[x=4,y=8, color=black]{a5}
\Vertex[x=4,y=6, color=black]{a6}
\Vertex[x=2,y=3]{t11}
\Vertex[x=4,y=2]{t12}
\Vertex[x=0,y=2]{t13}
\Vertex[x=6,y=8]{t51}
\Vertex[x=8,y=8]{t52}
\Vertex[x=10,y=9]{t53}
\Vertex[x=10,y=7]{t54}
\Vertex[x=8,y=6]{t55}
\Vertex[x=6,y=10]{t61}
\Vertex[x=8, y=10]{t62}
\Vertex[x=-2, y = 8]{t31}
\Vertex[x=-2, y = 7]{t32}
\Vertex[x=-2, y = 9]{t33}
\Vertex[x=-4, y = 9]{t34}
\Vertex[x=-2, y = 5]{t35}
\Vertex[x=-2, y = 3]{t36}

\draw (a3) -- (t31);
\draw (a3) -- (t32);
\draw (a3) -- (t33);
\draw (t33) -- (t34);
\draw (t32) -- (t35);
\draw (t35) -- (t36);
\draw (a1) -- (a2);
\draw (a2) -- (a3);
\draw (a3) -- (a4);
\draw (a4) -- (a5);
\draw (a5) -- (a6);
\draw (a6) -- (a1);

\draw (a1) -- (t11);
\draw (t12) -- (t11);
\draw (t13) -- (t11);

\draw (a5) -- (t51);
\draw (t51) -- (t52);
\draw (t52) -- (t53);
\draw (t52) -- (t54);
\draw (t52) -- (t55);

\draw (a5) -- (t61);
\draw (t61) -- (t62);
\Vertex[x=22,y=5, color=black]{b1}
\Vertex[x=20,y=6, color=black]{b2}
\Vertex[x=20,y=8, color=black]{b3}
\Vertex[x=22,y=9, color=black]{b4}
\Vertex[x=24,y=8, color=black]{b5}
\Vertex[x=24,y=6, color=black]{b6}
\Vertex[x=18, y = 8]{b31}
\Vertex[x=18, y = 7]{b32}
\Vertex[x=18, y = 9]{b33}

\Vertex[x=22, y = 3]{b11}
\Vertex[x=26, y = 8]{b51}
\Vertex[x=26, y = 10]{b52}

\draw (b3) -- (b31);
\draw (b3) -- (b32);
\draw (b3) -- (b33);
\draw (b1) -- (b11);
\draw (b5) -- (b51);
\draw (b5) -- (b52);
\draw (b1) -- (b2);
\draw (b2) -- (b3);
\draw (b3) -- (b4);
\draw (b4) -- (b5);
\draw (b5) -- (b6);
\draw (b6) -- (b1);

\Text[x=2,y=0] { \Huge $\rm(a)$};
\Text[x=22,y=0] { \Huge $\rm(b)$};
\end{tikzpicture}
\caption{Instances of graphs discussed in the proof of Lemma ~\ref{lem:max3}. 
(a) A cycle with trees attached to it. (b) A cycle with single edges
attached to it.}
\label{fig:cycle:trees}
\end{figure}

Suppose that, for some $i$, there is a vertex $x\in V(T_i)\setminus V(C)$ such that $\deg(x)\ge 2$. 
We may assume that $x$ has only one neighbor of degree $\ge2$ and $d \le k-1$ neighbors of degree $1$. Let $y$ be one of the leaf neighbors. We pick the edge $xy$ as a matching edge and remove $x$ and its leaf neighbors from $G$. What remains is a connected graph $G'$ of order $n-(d+1)$ that has only one cycle and maximum degree $\le k$. By induction hypothesis, we have $\nu(G')\ge \frac{n-(d+1)}k$. Therefore,
$\nu(G) \ge 1+ \frac{n-(d+1)}k = \frac{k}k + \frac{n-d-1}k \ge \frac{n}k$.

Now suppose there no such vertex $x$ with the conditions mentioned above. In this case, the trees $T_1, \ldots, T_r$ are single edges attached to the cycles $C$. See Figure \ref{fig:cycle:trees}\,(b).
We distinguish two cases. In one case $G$ is a simple cycle (the trees are empty). Every cycle of order $n$ has a matching of size at least $\frac{n-1}2 \ge \frac{n}k$. Next, suppose that there are some edges attached to the cycle $C$. 
Let $z$ be a vertex in the cycle $C$ that has a neighbor $y$ of degree $1$. We pick the edge $zy$ as a matching edge and remove $z$ and its leaf neighbors from $G$. Let $d \le k-1$ be the number of vertices removed. What is left is a tree $T'$ of order $n-d$ and maximum degree $\le k$. By the statement in the first part, we have $\nu(T') \ge \frac{n-d-1}k$. 
Therefore $\nu(G) \ge 1 + \frac{n-d-1}k = \frac{k}k + \frac{n-d-1}k \ge \frac{n}k$. This finishes the proof.
\end{proof}

We are now prepared to prove our main result. 

\begin{proof}[\bf Proof of Theorem \ref{thm:main}]
Since $\bind(X)=1/k$, by Lemma~\ref{lem:richness}, there exists a function $f:X \rightarrow N(X)$ with the properties described in the statement of the lemma. Let us denote  the image of $X$ under $f$ by $f(X)$.

Consider a directed graph $\vec{H}$ defined in the following way: $V(\vec{H}) = X \cup f(X)$ and there is an edge $(x,y) \in E(\vec{H})$ if  $f(x)=y$. We define the undirected version of $\vec{H}$, denoted by $H$, by simply dropping the directions of the edges in $\vec{H}$. Note that $H$ may have parallel edges.
 However, if we exclude these parallel edges, $H$ is indeed a subgraph of $G$. Therefore, any matching of $H$ is also a matching in $G$.
 Let $H_1, \ldots, H_t$ be the components of $H$ and let $\vec{H}_1, \ldots, \vec{H}_t$ be the directed counterparts of these components. We set $X_j := X \cap V(\vec{H}_j)$ for $j=1,\ldots, t$.
Let us fix an arbitrary $i \in \{1, \dots, t\}$. We claim that $H_i$ has a matching of size at least $|X_i|/(k+1)$. Since the components are vertex-disjoint and we have $|X_1|+\cdots+|X_t| = |X|$, proving this claim is sufficient to establish the theorem.
To prove our claim, consider the directed component $\vec{H}_i$. 

First assume that $\vec{H}_i$ is a DAG (directed acyclic graph). Every DAG has a vertex with out-degree zero. Let $u$ be such a vertex in $\vec{H}_i$. 
The vertex $u$ cannot be in $X$ by the definition of $\vec{H}$. Hence
$|X_i| \le |V(\vec{H}_i)| -1$. Note that since $\vec{H}_i$ is a DAG, the undirected version $H_i$ has no parallel edges. Moreover, its maximum degree is $k+1$. By Lemma \ref{lem:max3} and the previous observation, it follows that $\nu(H_i) \ge \frac{|V(H_i)| -1}{k+1} = \frac{|V(\vec{H}_i)| -1}{k+1} \ge \frac{|X_i|}{k+1}$.

Therefore, for the rest of the proof, we assume that $\vec{H}_i$ has a directed cycle. In this case, we will demonstrate that $H_i$ has a particular structure. Specifically, we will show that $H_i$ contains only one cycle (where we consider parallel edges as a cycle). To achieve this, we need to make a few observations.

\begin{observation} 
\label{obs:directedcycle} 
Let $ e_1,\ldots, e_r$ be a sequence of edges that forms a cycle $C$ in $H_i$. Let
$\vec{C}$ be the same sequence of edges with the directions restored. Then $\vec{C}$ must be a directed cycle. 
\end{observation}
\begin{proof} Note that $\vec{C}$ has at least two edges. If $\vec{C}$ is not a directed cycle, then there must be a vertex in $\vec{C}$
with two outgoing edges. This contradicts the definition of $\vec{H}$. 
\end{proof}
\begin{observation} 
\label{obs:directedpath} 
Let $C$ be a cycle in $H_i$ and  $x\in V(H_i)\setminus V(C)$. Let $P$
be a simple path that connects $x$ to $C$. Then all the edges in $\vec{P}$ must be directed
toward the cycle. 
\end{observation}
\begin{proof} Let $z$ be a vertex in $C$ that is an endpoint of $P$ in $H_i$. The other endpoint is $x$. By Observation~\ref{obs:directedcycle}, $\vec{C}$ must be a directed cycle, and since $z$ cannot have two outgoing edges, the edge $yz$ in $P$ that meets $C$ must point towards $z$. Consequently, all the edges in $\vec{P}$ must point toward $C$, otherwise 
we can find a vertex with two outgoing edges, which is a contradiction.
\end{proof}
Now we are ready to show that $H_i$ has only one cycle. Suppose $H_i$ has two distinct cycles $C_1$ and
$C_2$. We distinguish four cases:
\begin{itemize}
\item[(a)] $C_1$ and $C_2$ are vertex-disjoint. In this case, let $P$ be a simple path that 
connects $C_1$ to $C_2$ in $H_i$. Let $x$ be the starting vertex of $P$ that is in $V(C_1)$. See Figure \ref{fig:contradictions}\,(a). By Observation~\ref{obs:directedpath},the initial edge of $\vec{P}$ must leave from $x$. This indicates that $x$ has two outgoing edges in $\vec{H}_i$, which is a contradiction.
\item[(b)] $C_1$ and $C_2$ share a vertex $x$ but $E(C_1)\cap E(C_2) = \emptyset$. See Figure~\ref{fig:contradictions}\,(b). Again by Observation~\ref{obs:directedpath}, $x$ ends up having two out-going edges, which is a contradiction. 
\item[(c)] $V(C_1) = V(C_2)$ and $E(C_1)\cap E(C_2) \neq \emptyset$. This situation can arise only when $C_1$ and $C_2$ differ in only one parallel edge, as shown in Figure~\ref{fig:contradictions}\,(c). By Observation~\ref{obs:directedcycle}, we reach a contradiction because $\vec{C_1}$ and $\vec{C_2}$ cannot both be directed cycles.
\item[(d)] $V(C_1) \neq V(C_2)$ and $E(C_1)\cap E(C_2) \neq \emptyset$. Let us assume, without loss of generality, that there exists a vertex $x \in V(C_1) \setminus V(C_2)$.  There must be two disjoint paths $P_1$ and $P_2$ that connects $x$ to  $C_2$, as shown in Figure~\ref{fig:contradictions}\,(d). Beginning from $x$, these paths have different starting edges. By  Observation~\ref{obs:directedpath}, both $\vec{P}_1$ and $\vec{P}_2$ should be directed toward $C_2$. This implies that $x$ has two outgoing edges which is a contradiction. 
\end{itemize}
Therefore, we can conclude that $H_i$ contains only one cycle.

From Lemma~\ref{lem:max3}, we can deduce that $H_i$ has a matching of size at least $\frac{ |V(H_i)|}{k+1} \ge \frac{|X_i|}{k+1}$. This completes the proof of Theorem~\ref{thm:main}. 
\end{proof}

\begin{figure}
\centering
\begin{tikzpicture}[scale=0.45,transform shape]

\Vertex[x=1,y=12, color=black]{a1}
\Vertex[x=1,y=14, color=black]{a2}
\Vertex[x=3,y=14, color=black]{a3}
\Vertex[x=3,y=12, color=black]{a4}
\Vertex[x=1,y=5, color=black]{b1}
\Vertex[x=1,y=7, color=black]{b2}
\Vertex[x=3,y=7, color=black]{b3}
\Vertex[x=3,y=5, color=black]{b4}
\draw (a1) -- (a2);
\draw (a2) -- (a3);
\draw (a3) -- (a4);
\draw (a4) -- (a1);
\draw (b1) -- (b2);
\draw (b2) -- (b3);
\draw (b3) -- (b4);
\draw (b4) -- (b1);
\draw (b4) -- (a4);
\Text[x=2.5, y = 9]{\Huge $P$}
\Text[x=2, y = 13]{\Huge $C_1$}
\Text[x=2, y = 6]{\Huge $C_2$}
\Text[x=3.7, y = 12]{\Huge $x$}
\Text[x=2, y = 3]{\Huge $\rm(a)$}
\Text[x=8, y = 3]{\Huge $\rm(b)$}
\Text[x=17, y = 3]{\Huge $\rm(c)$}
\Text[x=27, y =3]{\Huge $\rm(d)$}
\Vertex[x=7,y=13, color=black]{c1}
\Vertex[x=9,y=14, color=black]{c2}
\Vertex[x=9,y=9, color=black]{c3}
\Vertex[x=7,y=6, color=black]{d1}
\Vertex[x=9,y=9, color=black]{d2}
\Vertex[x=9,y=5, color=black]{d3}
\draw (c1) -- (c2);
\draw (c2) -- (c3);
\draw (c3) -- (c1);
\draw (d1) -- (d2);
\draw (d2) -- (d3);
\draw (d3) -- (d1);
\draw (d3) -- (c3);
\Text[x=9.7, y = 9]{\Huge $x$}
\Text[x=8.3, y = 12]{\Huge $C_1$}
\Text[x=8.3, y = 6.5]{\Huge $C_2$}
\Vertex[x=17,y=8, color=black]{e1}
\Vertex[x=15,y=9, color=black]{e2}
\Vertex[x=15,y=11, color=black]{e3}

\Vertex[x=19,y=11, color=black]{e4}
\Vertex[x=19,y=9, color=black]{e5}
\draw (e1) -- (e2);
\draw (e2) -- (e3);
\Edge[Direct, lw=1pt, bend=30](e3)(e4)
\Edge[Direct, lw=1pt, bend=30](e4)(e3)
\draw (e4) -- (e5);
\draw (e5) -- (e1);
\Text[x=16.8, y = 9.5]{\Huge $C_1, C_2$}

\Vertex[x=26,y=8, color=black]{f1}
\Vertex[x=24,y=9, color=black]{f2}
\Vertex[x=24,y=11, color=black]{f3}
\Vertex[x=28,y=11, color=black]{f4}
\Vertex[x=28,y=9, color=black]{f5}
\Vertex[x=30,y=9, color=black]{f6}
\Vertex[x=30,y=11, color=black]{f7}
\draw (f1) -- (f2);
\draw (f2) -- (f3);
\draw (f3) -- (f4);
\draw (f4) -- (f5);
\draw (f5) -- (f1);
\draw (f5) -- (f6);
\draw (f4) -- (f7);
\draw (f7) -- (f6);
\Text[x=26, y = 10]{\Huge $C_1$}
\Text[x=29, y = 10]{\Huge $C_2$}
\Text[x=30.7, y = 11]{\Huge $x$}

\end{tikzpicture}
\caption{A pictorial representation of the contradictory cases in the proof of Theorem \ref{thm:main}.
}
\label{fig:contradictions}
\end{figure}
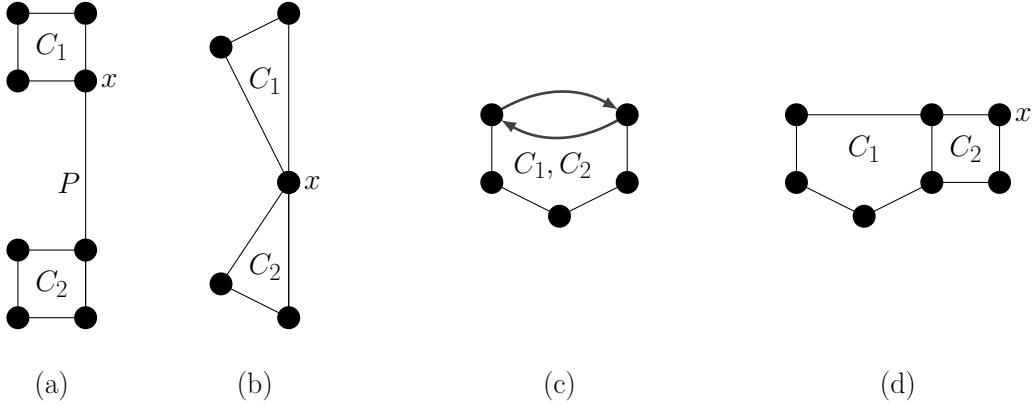

\section{Applications of Theorem \ref{thm:main}}
\label{sec:consequences}

In this section we demonstrate the applications of Theorem~\ref{thm:main}. 
The following fact is well-known and is used in the subsequent proofs. 

\begin{fact}
\label{fact}
Let $G=(V, E)$ be a graph. The following statements are true. 
\begin{itemize}
\item If $G$ is a bipartite planar graph, then $|E| \le 2|V|-4 $. 
\item If the arboricity of $G$ is bounded by $\alpha$, then $|E| \le \alpha|V|$. 
\end{itemize}
\end{fact}

As an application of Theorem \ref{thm:main}, we first prove the following lemma which is a common basis for proving both Theorems~\ref{thm:ls} and \ref{thm:shallowhigh}. For a graph $G$, 
recall that $L(G)$ is the set of locally superior vertices in $G$ and $H_t(G)$ is the set of vertices in $G$ with degree greater than $t$. 
We also define $L_t(G)$ as the set of locally superior vertices in $G$ with degree is at most $t$. In other words, $$L_t(G) = \{x \in L(G) : \deg(x) \le t\}.$$

\begin{lemma} 
\label{lem:KtG}
Let $t \ge 2$ be an integer and $K_t(G):= L_t(G) \cup H_t(G)$. The following statements are true. 
\begin{itemize} 
\item If $G$ is a planar graph, then $ |K_2(G)| \le 3\nu(G)$.
\item If $G$ is a $t$-bounded arborocity graph, then $ |K_t(G)| \le (t+1)\nu(G)$. 
\end{itemize}
\end{lemma}

\begin{proof} 
 Assuming that $G$ is a $t$-bounded arboricity graph or a planar graph, we show that $\bind(K_t(G)) \ge1/t$. Given this, the assertion follows from Theorem \ref{thm:main} 

Fix $t \ge 2$ and consider an arbitrary subset $Z \subseteq K_t(G)$. Let $Z' = Z \setminus N(Z)$. We show that $|N(Z')|\ge \frac1t |Z'|$. Consider the bipartite graph $G'=(Z'\cup N(Z'), E')$. We use a charging argument. If there is $y \in N(Z')$ that has at most $t$ neighbors in $Z'$, we assign the neighbors to  $y$. We remove $y$ and its neighbors from $G'$ and continue in this fashion until  no vertex is left in $N(Z')$  with at most $t$ neighbors in $Z'$. Let $G''=(Z'' \cup N(Z''), E'')$ be the remaining graph. First, we note that all the vertices in $N(Z'')$ have degree at least $t+ 1$. This implies that $Z''\subseteq H_t(G)$, otherwise if there is $x \in Z''\cap L_i(G)$ for some $i \le t$, it must have a neighbor in $N(Z'')$ with degree at most $i \le t$ which is a contradiction. Since $t \ge 2$ and $G''$ is a subgraph of a planar graph or a $t$-bounded arborocity graph, by Fact~\ref{fact}, we have $|E''| \le t (|N(Z'')| + |Z''|)$. Thus, we have $(t+1)|Z''| \le |E''| \le t (|N(Z'')| + |Z''|)$, which implies that $|N(Z'')| \ge \frac{1}{t} |Z''|$. Wrapping up we obtain that
\begin{eqnarray} 
|N(Z)| & = & |Z\setminus Z'| + |N(Z')| \hspace{3.3cm} \text{ (by the definition of $Z'$)} \\
& \ge & |Z\setminus Z'| +\frac1t |Z' \setminus Z''| + |N(Z'')| \hspace{1cm} \text{ (by the assignment procedure)}\\
& \ge & |Z\setminus Z'| +\frac1t|Z' \setminus Z''| + \frac1t |Z''| \\
& \ge & \frac1t |Z\setminus Z'| + \frac1t |Z' \setminus Z''| + \frac1t |Z''| \\
& \ge & \frac1t |(Z\setminus Z') \: \cup \: (Z' \setminus Z'') \: \cup \: Z''| = \frac1t |Z|.
\end{eqnarray}
This finishes the proof. 
\end{proof}

\vspace{0.5cm}

\begin{proof}[\bf Proof of Theorem~\ref{thm:ls}] The left-hand side of the inequalities follows from the fact that in every edge, at least one of the endpoints is a locally superior vertex. For the right-hand side, we note that $L(G) \subseteq K_t(G)$ holds true for any positive integer $t$. Therefore, $|L(G)| \le |K_t(G)|$, and by applying Lemma~\ref{lem:KtG}, we obtain the assertion.
\end{proof}

\vspace{0.5cm}

A randomized streaming algorithm
for approximating $\nu(G)$ in the vertex-arrival model using $O(\frac{\sqrt{n}}{\eps^2})$ space is  presented \cite{Jowhari23}. Relying on the facts that 
$|L(G)|\le 3.5\nu(G)$ for planar graphs and $|L(G)| \le (\alpha+2)\nu(G)$ for $\alpha$-bounded arboricity graphs (using a different analysis), the author argues that the approximation factor of the algorithm is respectively $3.5+\eps$ and $\alpha+2+\eps$ for planar graphs and $\alpha$-bounded arboricity graphs. As a consequence of  Theorem~\ref{thm:ls}, we can establish improved bounds for the approximation factor of the algorithm proposed in \cite{Jowhari23}. Therefore, we obtain the following algorithmic corollary.

\begin{corollary} Let $G$ be a planar graph of order $n$.
There is a randomized data stream algorithm in the vertex-arrival model that approximates
$\nu(G)$ within $3\pm \eps$ factor and uses $\tilde{O}(\frac{\sqrt{n}}{\eps^2})$ space. 
For a graph $G$ with arboricity bounded by $\alpha$, there is an algorithm for approximating $\nu(G)$ within $1+\alpha \pm \eps$ factor using the same space bound. 
\end{corollary}

We conclude this section by presenting the proof of our 3 factor approximation algorithm for finding the matching number of planar graphs.
\begin{proof}[\bf Proof of Theorem~\ref{thm:shallowhigh}]
Recall that $S_t(G) = \{ (u,v) \in E :\deg(u) \le t \text{ and } \deg(v) \le t\}.$
First we observe that $\nu(G) \le |S_t(G)|+|H_t(G)|$ is true. To see this, let 
$M$ be a maximum matching in $G$. Every edge $e \in M$ is either 
in $S_t(G)$ or one of its endpoints is in $H_t(G)$. This implies the left-hand side of the inequality. 
To prove the right-hand side, let $G$ be a planar graph. We need to show that 
$|S_2(G)|+|H_2(G)| \le 3 \nu(G)$. We proceed by  induction on the number of edges. Clearly, the assertion is true for an empty graph. Suppose $G$ has a shallow edge, i.e. there is an edge $uv \in S_2(G)$. In this case, we remove all the edges on $u$ and $v$ and obtain 
another graph $G'$. Since $\deg(u) \le 2$ and $\deg(v) \le 2$, by removing the edges on $u$ and $v$, the set $S_2(G) \cup H_2(G)$ loses at most $3$ elements. Note that removing edges may introduce some new shallow edges. Nonetheless, we can say
$|S_2(G)|+|H_2(G)| \le 3 + |S_2(G')|+|H_2(G')|$. By the induction hypothesis, we have $ |S_2(G')|+|H_2(G')| \le 3\nu(G') $. Therefore, we have $|S_2(G)|+|H_2(G)| \le 3+ 3\nu(G')$. On the other hand, $\nu(G) \ge \nu(G') + 1$. Hence, we obtain  $|S_2(G)|+|H_2(G)| \le 3 + 3(\nu(G)-1) \le 3\nu(G)$. 
Now, suppose that $G$ has no shallow edges, i.e. $S_2(G) = \emptyset$. By Lemma \ref{lem:KtG}, we have $|L_2(G) \cup H_2(G)| \le 3\nu(G)$. Clearly $|H_2(G)|\le 3\nu(G)$. Consequently, in this case, we get $|S_2(G)|+|H_2(G)|\le 3\nu(G)$. This completes the proof. \end{proof}

\section{Concluding Remarks}

The reader my wonder what is the rationale for excluding subsets $S$ where $N(S)=V(S)$ in the definition of the binding number. The reason for this exclusion is that it allows for a wider range of values for the binding number. Without this exclusion, the binding number would be limited to the range of $[0,1]$, whereas with the current definition, the binding number can be as high as $n-1$, which occurs in the complete graph $K_n$. Woodall \cite{Woodall73}, in particular, presented interesting results for graphs with binding number above $1$. For instance, he proved that every graph with binding number at least $\frac32$ has a Hamiltonian circuit. Additionally, note that for $X \subseteq V(G)$ with binding number in the range $[0,\frac12]$, we could drop this exclustion from the definition and without any harm. However in order to avoid making new definitions and for the sake of consistency, we are sticking with binding number as defined by Woodall. 

Theorem \ref{thm:main} is tight. As a simple witness, consider the star graph $S_k$ on $k+1$ vertices. The size of its maximum matching is $\frac{k+1}{k+1}=1$ while the binding number of $S_k$ is $\frac1k$. Also one can see that the statement does not follow in the case of $k=1$. As a counter-example consider the odd cylce $C_{n}$. The binding number of $C_n$ is $1$ while the size of matching in $C_n$ is less than $\frac{n}{2}$. 

The same proof provided for Theorem~\ref{thm:shallowhigh} also works for graphs with arboricity $2$. It is tempting to conjecture that $|S_t(G)|+|H_t(G)|\le (t+1)\nu(G)$ when $G$ is a graph has arboricity bounded by $t$. Unfortunately the inductive proof of Theorem~\ref{thm:shallowhigh} does not carry to the cases where $t \ge 3$. However, we know that when $t=1$ (i.e., $G$ is a forest), the statement $|S_1(G)|+|H_1(G)| \le 2\nu(G)$ is true. 
To see this, note that the isolated edges are part of any maximum matching. Also it is known \cite{BuryGMMSVZ19} that in a tree, the number of non-leaves is at most twice the size of the maximum matching. This proves the statement for the forests. 


\bibliographystyle{plainurl}
\bibliography{mybibfile}{}

\end{document}